\documentclass[a4paper,UKenglish,cleveref, autoref, thm-restate]{lipics-v2021}
\usepackage{multirow, siunitx}

\nolinenumbers
\def\equationautorefname~#1\null{%
  (#1)\null
}
\usepackage{graphicx}
\usepackage{color}
\usepackage{subcaption}
\usepackage{graphics}

\usepackage{enumerate}

\usepackage{amssymb}
\usepackage{amsmath}
\usepackage{amsfonts}
\usepackage{bm}
\usepackage{mathtools}
\usepackage{cuted}
\usepackage{verbatim}

\usepackage{textcomp}

\usepackage{multirow}

\usepackage{booktabs}

\usepackage{url}

\usepackage{balance}

\usepackage[lined,boxed,linesnumbered,vlined]{algorithm2e}
\usepackage{algpseudocode}

\usepackage{etoolbox}

\def\equationautorefname~#1\null{%
  (#1)\null
}

\SetKwInput{Problem}{problem}
\SetKwInput{Input}{input}
\SetKwInput{Output}{output}
\SetAlgoSkip{smallskip} 
\SetAlgoInsideSkip{smallskip} %
\SetKwComment{tcp}{\#\ }{}
\SetKwComment{tcc}{\#\ }{}
    
\usepackage{amsthm}
\newtheoremstyle{saber}
  {}
  {}
  {\itshape}
  {}
  {\bfseries}
  {.}
  {.5em}
  {}
\theoremstyle{saber}
\ifcsundef{theorem}{
\newtheorem{theorem}{Theorem}[section]
}{}
\ifcsundef{corollary}{%
\newtheorem{corollary}{Corollary}[section]
}{}
\ifcsundef{lemma}{%
\newtheorem{lemma}{Lemma}[section]
}{}
\ifcsundef{definition}{%
\newtheorem{definition}{Definition}[section]
}{}
\ifcsundef{fact}{%

}{}
\ifcsundef{remark}{%
}{}
\ifcsundef{property}{%
}{}

\ifx\theoremautorefname\undefined
\newcommand*{\theoremautorefname}{Theorem}
\fi
\ifx\lemmaautorefname\undefined
\newcommand*{\lemmaautorefname}{Lemma}
\fi
\ifx\definitionautorefname\undefined
\newcommand*{\definitionautorefname}{Definition}
\fi
\ifx\corollaryautorefname\undefined
\newcommand*{\corollaryautorefname}{Corollary}
\fi
\ifx\factautorefname\undefined
\newcommand*{\factautorefname}{Fact}
\fi
\ifx\propertyautorefname\undefined
\newcommand*{\propertyautorefname}{Property}
\fi

\ifx\argmax\undefined
\newcommand{\argmax}{\operatornamewithlimits{argmax}}
\fi
\ifx\argmin\undefined
\newcommand{\argmin}{\operatornamewithlimits{argmin}}
\fi
\ifx\limsup\undefined
\newcommand{\limsup}{\operatornamewithlimits{limsup}}
\fi
\ifx\liminf\undefined
\newcommand{\liminf}{\operatornamewithlimits{liminf}}
\fi
\ifx\norm\undefined
\newcommand{\norm}[1]{\lVert#1\rVert}
\fi
\ifx\abs\undefined
\newcommand{\abs}[1]{\lvert#1\rvert}
\fi
\ifx\set\undefined
\newcommand{\set}[1]{\left\{#1\right\}}
\fi
\ifx\mset\undefined
\newcommand{\mset}[1]{\lbrack #1\rbrack}
\fi
\ifx\etal\undefined
\newcommand{\etal}[1]{{\em #1 et al.}~}
\fi
\ifx\ie\undefined
\newcommand{\ie}{{\em i.e.,} }
\fi
\ifx\eg\undefined
\newcommand{\eg}{{\em e.g.,} }
\fi
\ifx\etc\undefined
\newcommand{\etc}{{\em etc.,} }
\fi
\ifx\wrt\undefined
\newcommand{\wrt}{{\em w.r.t.} }
\fi
\ifx\dotprod\undefined
\newcommand{\dotprod}[2]{
  \langle #1, #2 \rangle
}
\fi
\ifx\iid\undefined 
\newcommand{\iid}{i.i.d.}
\fi
\ifx\bigParenthes\undefined 
\newcommand{\bigParenthes}[1]{
  \big(#1\big)
}
\fi
\ifx\bigBracket\undefined 
\newcommand{\bigBracket}[1]{
  \big\{#1\big\}
}
\fi
\ifx\bigSqBracket\undefined 
\newcommand{\bigSqBracket}[1]{
  \big[#1\big]
}
\fi
\ifx\BigParenthes\undefined 
\newcommand{\BigParenthes}[1]{
  \Big(#1\Big)
}
\fi
\ifx\BigBracket\undefined 
\newcommand{\BigBracket}[1]{
  \Big\{#1\Big\}
}
\fi
\ifx\BigSqBracket\undefined 
\newcommand{\BigSqBracket}[1]{
  \Big[#1\Big]
}
\fi
\ifx\biggParenthes\undefined 
\newcommand{\biggParenthes}[1]{
  \bigg(#1\bigg)
}
\fi
\ifx\biggBracket\undefined 
\newcommand{\biggBracket}[1]{
  \bigg\{#1\bigg\}
}
\fi
\ifx\biggSqBracket\undefined 
\newcommand{\biggSqBracket}[1]{
  \bigg[#1\bigg]
}
\fi
\ifx\BiggParenthes\undefined 
\newcommand{\BiggParenthes}[1]{
  \Bigg(#1\Bigg)
}
\fi
\ifx\BiggBracket\undefined 
\newcommand{\BiggBracket}[1]{
  \Bigg\{#1\Bigg\}
}
\fi
\ifx\BiggSqBracket\undefined 
\newcommand{\BiggSqBracket}[1]{
  \Bigg[#1\Bigg]
}
\fi
\ifx\bracket\undefined
\newcommand{\bracket}[1]{
  \{#1\}
}
\fi
\ifx\parenthes\undefined
\newcommand{\parenthes}[1]{
  (#1)
}
\fi
\ifx\sqBracket\undefined
\newcommand{\sqBracket}[1]{
  [#1]
}
\fi
\ifx\prob\undefined 
\newcommand{\prob}[1]{\mathbb{P}[#1]}
\fi
\ifx\Prob\undefined 
\newcommand{\Prob}[1]{\mathbb{P}\big[#1\big]}
\fi
\ifx\probb\undefined 

\fi
\ifx\expect\undefined 
\newcommand{\expect}[1]{\mathbb{E}[#1]}
\fi
\ifx\Expect\undefined 
\newcommand{\Expect}[1]{\mathbb{E}\big[#1\big]}
\fi
\ifx\expectt\undefined 
\newcommand{\expectt}[1]{\mathbb{E}\bigg[#1\bigg]}
\fi
\ifx\walk\undefined 
\makeatletter
\newcommand{\walk}[1]{%
  \@tempswafalse
  \@for\next:=#1\do
    {\if@tempswa\!\!\rightarrow\!\!\else\@tempswatrue\fi\next}%
}
\makeatother
\fi
\ifx\seq\undefined 
\newcommand{\seq}{\!=\!}
\fi
\ifx\sminus\undefined 
\newcommand{\sminus}{\!-\!}
\fi
\ifx\sm\undefined 
\newcommand{\sm}[1]{\!#1\!}
\fi

\ifx\union\undefined
\newcommand{\union}[2]{#1\!\cup\!#2}
\fi

\ifx\hl\undefined 
\newcommand{\hl}[2]{{\color{#1}#2}}
\fi

\ifx\hlb\undefined 
\newcommand{\hlb}[1]{{\color{blue}#1}}
\fi

\ifx\hlr\undefined 
\newcommand{\hlr}[1]{{\color{red}#1}}
\fi

\ifx\hlp\undefined 
\newcommand{\hlp}[1]{{\color{purple}#1}}
\fi

\ifx\hlc\undefined 
\newcommand{\hlc}[3]{{\color{#1}#2 [remark: #3]}}
\fi

\ifcsdef{sitemize}{}{

}

\ifcsdef{senumerate}{}{

}

\ifcsdef{spenumerate}{}{

}

\ifcsdef{slenumerate}{}{

}

\title{A Dyadic Simulation Approach to Efficient Range-Summability}
\author{Jingfan Meng}{School of Computer Science, Georgia Institute of Technology, USA}{jmeng40@gatech.edu}{}{}
\author{Huayi Wang}{School of Computer Science, Georgia Institute of Technology, USA}{huayiwang@gatech.edu}{}{}
\author{Jun Xu}{School of Computer Science, Georgia Institute of Technology, USA}{jx@cc.gatech.edu}{}{}
\author{Mitsunori Ogihara}{Department of Computer Science, University of Miami, USA}{ogihara@cs.miami.edu}{}{}

\begin{document}


\authorrunning{J. Meng and H. Wang and J. Xu and M. Ogihara} 

\Copyright{Jingfan Meng, Huayi Wang, Jun Xu, and Mitsunori Ogihara} 



\relatedversion{}
\funding{}
\acknowledgements{}

\maketitle
\begin{abstract}
Efficient range-summability (ERS) of a long list of random variables is a fundamental algorithmic problem that has applications to three important database applications, namely, data stream processing, space-efficient histogram maintenance (SEHM),  and approximate nearest neighbor searches (ANNS). In this work, we propose a novel dyadic simulation framework and develop three novel ERS solutions, namely Gaussian-dyadic simulation tree (DST), Cauchy-DST and Random Walk-DST, using it. We also propose novel rejection sampling techniques to make these solutions computationally efficient. Furthermore, we develop a novel $k$-wise independence theory that allows our ERS solutions to have both high computational efficiencies and strong provable independence guarantees. 
\end{abstract}


\section{Introduction}

In this work, we propose {\it dyadic simulation}, a novel solution framework to a fundamental algorithmic problem that has applications to three important database applications:  data stream processing, space-efficient histogram maintenance (SEHM) and approximate nearest neighbor searches (ANNS).
This algorithmic
problem, called \emph{efficient range-summability} (ERS) of random variables (RVs)~\cite{feigenbaum, rusu-jour}, can be stated as follows.
Let $X_0, X_1, \cdots, X_{U-1}$ be a list of i.i.d. RVs, where the (index) universe size $U$ is typically a large number (say $U = 2^{64}$).  
Given a range $[a, b) \triangleq \{a, a+1, \cdots, b-1\}$ that lies in $[0, U)$, 
we need to
compute $S[a, b) \triangleq \sum_{i=a}^{b-1}X_i$, the sum of the RVs $X_a, X_{a+1}, \cdots, X_{b-1}$ in the range.
A straightforward but naive solution to this problem, which follows an intuitive ``bottom-up'' approach, is to generate RVs 
$X_a, X_{a+1}, \cdots,$ $X_{b-1}$ individually and then add them up.
This solution, however, has a time complexity of
$O(b - a)$, which is inefficient computationally when the range length $b - a$ is large.   In contrast, an acceptable solution~\cite{feigenbaum, rusu-jour} 
should be able to do so with only O(polylog($b-a$)) time complexity.  


\subsection{Our Dyadic Simulation Approach}\label{sec:dyaappoach}

For ease of presentation, we make two harmless simplifying assumptions.   
The first assumption is a typical ``computer science'' one:  The universe size $U$ is a power of 2.  This assumption can always be fulfilled by increasing $U$ to at most $2U$. The second assumption is that 
$[a, b)$ is a dyadic range in the sense there exist integers $j \ge 0$ and $i \ge 0$ such that $a = j \cdot 2^i$ and $b = (j+1) \cdot 2^i$.  
It suffices for our solution to work for any dyadic range since any non-dyadic range can be split into at most $2\log_2 U$ dyadic ranges,
as we will elaborate in \autoref{sec:dyadicframework}.

Unlike the naive solution, our dyadic simulation approach computes $S[a, b)$ in a counter-intuitive ``top-down'' manner as follows.  
Its first step is to generate the RV $S[0, U)$, the range-sum of the entire universe.
If we denote the distribution of each underlying RV $X_i$ as $X$, 
then $S[0, U)$ has distribution $X^{*U}$, 
where, for any $n > 1$, $X^{*n}$ denotes the $n^{th}$ convolutional power of $X$.
When $X$ is one of a few nice distributions, the distribution $X^{*U}$ can be analytically derived and also takes a nice form;
in this case, it is straightforward to generate $S[0, U)$.
For example, when $X$ is standard Gaussian distribution $\mathcal{N}(0, 1)$, then $X^{*U}$ is $\mathcal{N}(0, U)$. 

The rest of dyadic simulation proceeds as follows.
If $[a, b)$ is the same as $[0, U)$, then the ERS problem is solved.
Otherwise, we split $S[0, U)$ into two half-range-sums $S[0, U/2) + S[U/2, U)$, such that RVs $S[0, U/2)$ and $S[U/2, U)$
are (mutually) independent and each has distribution $X^{*(U/2)}$.   While this may sound wishful thinking, we will show in~\autoref{sec:dyasim} that 
it is always mathematically possible and can be done in a computationally efficient manner in some cases.


After the split, we have either $[a, b) \subseteq [0, U/2)$ or $[a, b) \subseteq [U/2, U)$, since $[a, b)$, $[0, U/2)$, and $[U/2, U)$ are
all dyadic intervals.  
We then recursively ``binary-search'' for $[a, b)$ either in the left-half $[0, U/2)$ if $[a, b) \subseteq [0, U/2)$ or in the right-half $[U/2, U)$ if
$[a, b) \subseteq [U/2, U)$.
It is not hard to 
verify that after at most $\log_2 U$ such splits we can ``find'' $[a, b)$ and as a result compute $S[a, b)$.  
Hence, the time complexity of a dyadic simulation algorithm is 
$O(\log U)$ splits for generating any dyadic range-sum.   Perhaps surprisingly, even for generating any range-sum that is not necessarily dyadic,
the time complexity remains $O(\log U)$ splits instead of becoming $O(\log^2 U)$, as we will show in~\autoref{sec:dyadicframework}.

We can generate any underlying RV $X_i$ via $\log_2 U$ such binary splits, 
because $X_i \equiv S[i, i+1)$, and $[i, i+1)$ is a dyadic range.  We say dyadic simulation takes a ``top-down'' approach 
because when all the underlying RVs $X_0, X_1, \cdots, X_{U-1}$ are generated this way, they become the ``leaves'' (at the ``bottom'') of 
the complete binary tree of the binary splits involved in generating them.  This tree, called \emph{dyadic simulation tree} (DST), will be officially introduced in~\autoref{sec:dyadicframework}.  
In this work, we propose novel DST-based solutions to three ERS problems 
whose underlying RVs have Gaussian, Cauchy, and single-step random walk (RW) (\emph{aka}. Rademacher) distributions, respectively.   
We also propose novel rejection sampling techniques that make these three solutions, called Gaussian-DST, Cauchy-DST, and RW-DST respectively, computationally efficient.
Each binary split operation takes only
nanoseconds for Gaussian and 20+ nanoseconds for Cauchy and random walk.



All existing ERS solutions were proposed for the single-step random walk distribution $\Pr[X = 1] = \Pr[X = -1] = 0.5$.  
Here we highlight a key difference between our dyadic simulation approach and these ERS solutions.
This difference is a major contribution of this work.  
The underlying RVs $X_0$, $X_1$, $\cdots$, $X_{U-1}$ generated by our dyadic simulation approach 
are at least empirically independent for all practical purposes.
In contrast, those generated by all existing ERS solutions are strongly correlated.
For example, in the EH3 scheme proposed in~\cite{feigenbaum, rusu-jour}, 
the underlying RVs are approximately 4-wise independent, but all independence beyond 4-wise is completely destroyed.  
However, in nearly all applications of dyadic simulation that we will describe next, we need these RVs to be at least empirically independent.



A very sketchy idea of dyadic simulation was proposed, in a few
sentences, in a theory paper~\cite{dyahist} that mainly focused on the aforementioned SEHM problem.   
Although it was stated in~\cite{dyahist} that dyadic simulation can possibly be used for the ERS of Gaussian and Cauchy RVs, 
no computationally efficient technique was specified in it for binary-splitting a Gaussian or Cauchy RV, as will be elaborated in Subsections~\ref{sec:gauss} and~\ref{sec:cauchy}.








\subsection{Independence Guarantees}\label{sec:indepguarantee}


 

As we have just explained, each non-leaf node in a DST corresponds to a dyadic range $[a, b)$, whose two children correspond
to the two dyadic half ranges $[a, (a+b)/2)$ and $[(a+b)/2, b)$.  
We will show in~\autoref{sec:erssolutions} that each such non-leaf node, now identified by its corresponding dyadic range say $[a, b)$, 
is associated with a uniformly random binary string $C_{[a, b)}$ that determines
the values of half-range-sums $S[a, (a+b)/2)$ and $S[(a+b)/2, b)$ that the range-sum $S[a, b)$ is split into.   
Depending on how each $C_{[a, b)}$ is generated, we can obtain various theoretical guarantees concerning how independent the 
underlying RVs $X_0, X_1, \cdots, X_{U-1}$ are.

Ideally, each such $C_{[a, b)}$ should be a freshly generated RV in the sense that it is independent of all other RVs.  If this is the case, then we can prove that,
starting with $S[0, U)$ that is distributed as $X^{*U}$, 
the $U$ underlying RVs generated through these binary splits are i.i.d. with distribution $X$.  However, this idealized case is impractical when the universe
size $U$ is massive, since 
the value of each fresh generated RV would all have to be remembered (stored in memory) and there can be a massive number of them.   
In practice, we typically generate each such $C_{[a, b)}$ value (on demand) by applying a hash 
function $h(\cdot)$ to the dyadic range $[a, b)$.  There are two standard choices of such hash functions in the literature.   
The practical type is ``off-the-shelf'' random hash functions that can produce a hash value in nanoseconds, such as \texttt{wyhash}~\cite{wyhash}.
Although they provide no theoretical guarantees, they were demonstrated to ensure a level of empirical independence that is good enough for all practical applications~\cite{smhasher}.
The theoretical type, called $k$-wise independent hash functions~\cite{univhash,tabulation4wise,tabulationpower}, generates $(C_{[a,b)})$'s that are $k$-wise independent.
In this work, we establish a novel $k$-wise independence theory for DST which shows, among other things, that $k$-wise independence among $C_{[a, b)}$ values implies $k$-wise independence among the underlying RVs.   
Although the latter theoretical guarantee is weaker than the ideal all-wise mutual independence, it leads to rigorous theoretical guarantees that are strong enough 
for most ERS applications.

We note all our DST solutions can use Nisan's pseudorandom generator (PRG)~\cite{nisan},
which delivers strong independence guarantees for memory- (state-space-) constrained algorithms. 
However, Nisan's PRG is quite computationally intensive, and hence has never been implemented and used in practice. 
Indeed, a key contribution of our $k$-wise independence theory lies in its ability to satisfy the ``theoretical needs'' of most ERS applications 
using $k$-wise independent hash functions that are much less computationally intensive.

\subsection{Applications}\label{sec:application}
In this section, we describe the three aforementioned applications that motivate our DST-based ERS solutions.
Since we claim none of them as a contribution of this work, each description here is only detailed enough to explain how an ERS problem arises in it.
Furthermore, we will not elaborate on any application in the rest of this paper.

The \emph{first} application is data stream processing, where two of our ERS
	solutions extend an existing data streaming algorithm suite for efficiently handling range-updates. We start our introduction with an oversimplified characterization of the data stream model.
In this model, the precise system state
is comprised of a large number (say $U$) of counters $\sigma_0$, $\sigma_1$, $\cdots$, $\sigma_{U-1}$ whose values are initialized to 0.  A data stream is comprised of a large
number of data items that can take one of the following two forms:  {\it standard (point-update)} and {\it range-update}. In a standard data stream,
each item, say the $t^{th}$, in 
the data stream is in the form $(i_t, \delta_t)$.  This data item should cause the following update to the precise system state:   Counter $\sigma_{i_t}$ is to be incremented by $\delta_t$, which we call a {\it point update}. In a range-update data stream, which is more general (than standard data streams), each data item is in the form $\allowbreak([a_t, b_t), \delta_t)$.  In this case, for each index $i$ in the range $[a_t, b_t)$, 
the corresponding counter $\sigma_i$ needs to be incremented by 
$\delta_t$, which we call a {\it range update}.
A typical data streaming query is to estimate 
a certain function of the counter values after the updates caused by all the data items in the data stream are committed to the system state.  
For example, the $L_2$-norm and the $L_1$-norm estimation problems are
to estimate the values of $d_2 \triangleq (\sum_{i=0}^{U-1}|\sigma_i|^2)^{1/2}$ 
(the $L_2$-norm of the system state) and $d_1 \triangleq \sum_{i=0}^{U-1}|\sigma_i|$ (the $L_1$-norm), respectively. 
Since $U$ is usually too huge for the precise system state to fit in fast memory, a data streaming algorithm has to summarize it into a synopsis data structure called a \emph{sketch}, whose size is much smaller than $O(U)$.



A data streaming algorithm suite, proposed in~\cite{stabledist}, solves the $L_2$- and the $L_1$-norm estimation problems for standard data streams.
It employs a Gaussian sum or Cauchy sum sketch comprised of $r>0$ i.i.d. accumulators (viewed as RVs) $A_1, A_2, \cdots, A_r$.  
Since these accumulators are independent and functionally equivalent, 
it suffices to describe the point-update procedure for one such accumulator, 
which we denote as $A$.  $A$ is initialized to $0$ at the beginning.  
Given a point update $(i_t, \delta_t)$, $A$ is incremented by $\delta_t  X_{i_t}$, where $X_{i_t}$ is 
a standard Gaussian (for $L_2$-norm) or Cauchy (for $L_1$-norm) RV that is fixed after being generated on-demand
for the first time and is associated with 
the counter $\sigma_{i_t}$. 
After the entire data stream has passed, it was shown in~\cite{stabledist} that $A = \sum_{i=0}^{U-1}\sigma_i X_i$ is distributed as a Gaussian RV $\mathcal{N}(0, d_2^2)$ or a Cauchy RV 
$\mathrm{Cauchy}(0, d_1)$, wherein the parameters $d_2^2$ and $d_1$ can be estimated using standard estimators.

This algorithm can handle a range update $\allowbreak([a_t, b_t), \delta_t)$ as follows:
\begin{center}
\fbox{  \begin{minipage}{0.45\textwidth} \centering For $ i = a_t$ to $b_t- 1$, do $A\gets A + \delta_t  X_i$. \end{minipage}
}
\end{center}
However, the time complexity of this update procedure is $O(b_t - a_t)$, 
which is very high when $b_t - a_t$ is gigantic.  
In comparison, our Gaussian-DST and Cauchy-DST solutions can process this range update in $O(\log(b_t - a_t))$ time, 
since the net effect of this range update is to increment $A$ by $\delta_t \cdot (\sum_{i=a_t}^{b_t-1} X_i)$, which is precisely $\delta_t$ times the (Gaussian or Cauchy) 
range-sum $S[a_t, b_t)$.

The \emph{second} application is the space-efficient histogram maintenance (SEHM) problem in the data streaming setting, which as mentioned earlier was
the focus of~\cite{dyahist}.  
The precise system state to be approximately maintained by a proposed SEHM solution is a scaled
probability mass function (pmf) $f(\cdot)$ whose domain is the set of integers $\{0, 1, 2, \cdots, U-1\}$, where the universe $U$ is typically a large (positive) integer;
we denote this domain simply as $[0, U)$.   This $f(\cdot)$ starts as a zero function, and at any moment $\tau$, $f(\cdot)$ is defined by a stream of point updates before or at $\tau$
in the sense each point update $(i_\tau, \delta_\tau)$ causes the value of $f(i_\tau)$ to be incremented by $ \delta_\tau$.   Hence $f(\cdot)$ is a ``pmf in motion''.

A part of the SEHM problem is to answer the following query.  At any given moment $\tau$, the proposed SEHM solution needs to approximately
represent the snapshot of $f(\cdot)$ at $\tau$ using a good and simple \emph{histogram} function whose domain is also $[0, U)$.  
Here, a histogram $\mathbf{H}$ is a piecewise-constant function defined by $B$ non-overlapping intervals (buckets) $I_1, I_2, \cdots, I_B$ that comprise $[0, U)$ and
$B$ \emph{spline parameters} $\chi_1, \chi_2, \cdots, \chi_B$ that define the height of each bucket, as follows:  $\mathbf{H}(i) = \chi_j$ 
when $i\in I_j$, for $i = 0, 1, \cdots, U-1$.   The approximation error of $\mathbf{H}$ (relative to $f(\cdot)$) is defined as the $L_2$-error $(\sum_{i=0}^{U-1}|\mathbf{H}(i) - f(i)|^2)^{1/2}$ 
or the $L_1$-error $\sum_{i=0}^{U-1}|\mathbf{H}(i) - f(i)|$.  A histogram $\mathbf{H}$ is called simple when $B$ is small and called good when its approximate error is small.

A subproblem of this query problem is, given a (simple) candidate histogram $\mathbf{H}$, to determine whether it is good in terms of $L_2$- or $L_1$-error.
It was shown in~\cite{dyahist} that the SEHM problem can be solved by maintaining a Gaussian-sum (for the $L_2$ case) or a Cauchy-sum (for the $L_1$ case) sketch of $f(\cdot)$.
In addition, for solving this subproblem given a candidate histogram $\mathbf{H}$, a Gaussian-sum or Cauchy-sum sketch of $\mathbf{H}$ needs to be computed.
Suppose $\mathbf{H}(i) = \chi_j$ 
when $i\in I_j$ and $I_j = [a_j, a_{j+1})$.  Then
the value of an accumulator $A$ in the sketch of $\mathbf{H}$ takes value $A = \sum_{j=1}^{B}\chi_j S[a_j, a_{j+1})$ (as explained above), where each $S[a_j, a_{j+1})$
is a Gaussian or Cauchy {\it range-sum} that needs to be {\it efficiently computed}.
It was shown in~\cite{dyahist} that the $L_2$- or $L_1$-error of approximating $f(\cdot)$ by $\mathbf{H}$ can be estimated from the
difference between the sketches of $f(\cdot)$ and $\mathbf{H}$.

We now shift our attention to the third 
application of ERS:   Locality-Sensitive Hashing (LSH) schemes for approximate nearest neighbors searches (ANNS).  An 
ERS problem arises in efficiently implementing a state-of-the-art LSH solution, called multi-probe random-walk LSH 
(MP-RW-LSH)~\cite{mprwlsh}, for 
ANNS in Manhattan ($L_1$) distance.  As explained in~\cite{mprwlsh}, to compute the value of a random-walk LSH (RW-LSH) function acting on a query vector (as its argument), we need to map an 
(arbitrarily) given
nonnegative even integer $\phi$ to a $\phi$-step random walk.   Since $\phi$ can be very large, this computation is precisely an ERS problem with $X$ being a single-step random walk.
The aforementioned EH3 scheme~\cite{feigenbaum} does not work for this ERS problem for the following reason. 
It was shown in~\cite{mprwlsh} that, for MP-RW-LSH to work properly, the probability distribution of any computed range-sum $S[a, b)$ must be either identical or close to that of a $(b-a)$-step random walk.
This requirement, however, is not generally satisfied by EH3, which destroys all independence beyond $4$-wise.
In contrast, according to~\autoref{cor:stdist} (in~\autoref{sec:dstimpl}), our Random Walk (RW)-DST solution strictly satisfies this requirement when it is implemented using $2$-wise independent hash functions.

In this work, we make two major and nontrivial contributions.  First, we propose a dyadic simulation framework and develop three novel and computationally  efficient   ERS solutions, namely 
Gaussian-DST, Cauchy-DST and RW-DST, based on it.  
Second, we establish a novel $k$-wise independence theory that allows our ERS solutions to have both strong provable
independence guarantees and low computational complexities.

\section{Dyadic Simulation Theory}\label{sec:dyasim}

In this section, we first describe how to generate an arbitrary dyadic range-sum using a \emph{dyadic simulation tree} (DST)
of binary splits.  After that, we describe
three aforementioned DST-based \emph{efficient range-summability} (ERS) solutions for three different target distributions. These three solutions, called Gaussian-DST, Cauchy-DST, and RW-DST 
(RW for random walk) respectively, follow a common framework and differ only in the binary split procedure.
In the rest of the paper, whenever possible, we focus on the design and the efficient implementation of only a single instance of DST.  
A real-world application usually needs to use many DST instances~\cite{stabledist, dyahist, mprwlsh}.
These DST instances are 
independent in the sense that the full vector of underlying RVs $ X_0, X_1, \cdots, X_{U-1}$ generated by them are independent.  

Before we describe the dyadic simulation approach, we state the precise problem statement of ERS, which consists of three requirements.
First, the underlying RVs $X_0, X_1, \cdots, X_{U-1}$ 
are i.i.d. with distribution $X$.
Second, every range-sum $S[a, b)$ is equal to $X_a + X_{a+1} + \cdots + X_{b-1}$.
Third, given any range $[a, b)$, its range-sum $S[a, b)$ can be computed in $O(polylog(b - a))$ time.
Whereas the second and the third requirements are straightforward to satisfy, to {\it provably} satisfy the strict independence part of the first requirement, 
we have to make an idealized assumption that we will elaborate in  \autoref{sec:dyadicframework}.

As mentioned earlier, each range in $[0, U)$ can be partitioned into disjoint dyadic ranges.
Such a partitioning can usually be done in multiple ways, but only one such way 
results in the minimum number of partitions.
This minimum partitioning is called the \emph{dyadic cover}, which contains at most $O(\log U)$ dyadic ranges~\cite{rusu-jour}.
For example, the dyadic cover of $[4,11)$ contains three dyadic ranges:  $[4,8)$, $[8, 10)$, and $[10, 11)$.
In the rest of the paper, we only show how to compute the range-sum for a dyadic range, since the range-sum of an arbitrary range $[a, b)$
is the sum of the range-sums of the dyadic ranges in the dyadic cover of $[a, b)$.  
Also as explained earlier, for notational convenience and ease 
of presentation, we assume that the universe range $U$ is a power of $2$.

\subsection{Dyadic Simulation Framework}\label{sec:dyadicframework}

In this section, we describe the dyadic simulation framework, and 
prove that a DST-based ERS solution satisfies all three requirements specified earlier.
We illustrate a DST using a ``small universe'' example (with $U$ = 16) shown in \autoref{fig:prefix}.  Sitting at the root of the tree is the $S[0, 16)$,
which has distribution $X^{*16}$ by initialization.  
Its two children are the two half-range-sums $S[0, 8)$ and $S[8, 16)$ resulting from splitting $S[0, 16)$,
its four grandchildren are the four 
quarter-range-sums $S[0, 4)$, $S[4, 8)$, $S[8, 12)$ and $S[12, 16)$ resulting from splitting $S[0, 8)$ and $S[8, 16)$ respectively, and so on.   
At the bottom of the tree are the sixteen underlying RVs 
$X_0$, $X_1$, $\cdots$, $X_{15}$.  

\begin{figure}
	\centering
	\includegraphics[width=\textwidth]{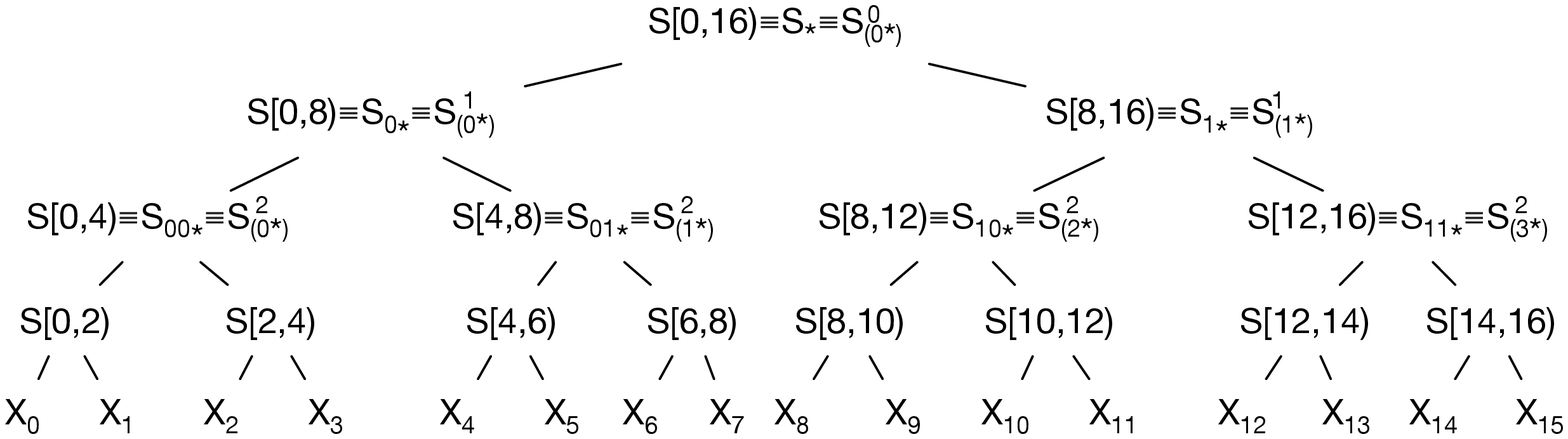}
	\caption{An illustration of the DST. }\label{fig:prefix}
\end{figure}

Under this modeling, every dyadic range-sum, including every underlying RV, 
corresponds to a node in this tree and its value is generated by binary-splitting all its ancestors.  
The computational complexity of generating a dyadic range-sum is clearly $O(\log U)$ splits.
~\autoref{cor:2logntime} states the aforementioned surprising result that the computational complexity of generating the sum of any range  
(that is not necessarily dyadic) is also $O(\log U)$ splits. Hence a DST-based solution satisfies the third requirement above.
The remark after the proof of~\autoref{th:twowise} gives an informal proof of~\autoref{cor:2logntime}.
In addition, under this dyadic simulation framework, the (dyadic) range-sum value of each non-leaf tree node is the sum of its two children.  
As a result, every dyadic range-sum $S[a, b)$ computed this way is indeed equal to $X_a +X_{a+1} + \cdots + X_{b-1}$.  Hence the second requirement
above is satisfied.

\begin{corollary}\label{cor:2logntime}
	For any integers $a, b$ such that $0 \leq a \leq b \leq U$, the range-sum $S[a, b)$ can be computed in no more than $2\log_2 U$ splits.
\end{corollary}


We now introduce the
concept of \emph{prefix} that will simplify our presentation next.  Viewing the DST as a binary trie, we can index each tree node as a prefix.  For example, 
in~\autoref{fig:prefix}, 
the range $[4, 8)$ is equivalent to the prefix $01*$ since it contains four binary numbers $4 = (0100)_2$, $5 = (0101)_2$, $6 = (0110)_2$, $7 = (0111)_2$ that 
share the common prefix 01.   



Next, we will prove that our DST-based approach satisfies the first requirement (underlying RVs being i.i.d.) above 
if the split procedure possesses two properties that we call (I) and (II).  Suppose
a dyadic range-sum $S_{\alpha*}$ that has distribution $X^{*2n}$ is split into $S_{\alpha0*} + S_{\alpha1*}$.  
\emph{Property (I)} is that 
$S_{\alpha0*}$ and $S_{\alpha1*}$
are i.i.d. with distribution $X^{*n}$. 
\emph{Property (II)} is that the random vector $\langle S_{\alpha0*}, S_{\alpha1*}\rangle$ is a (vector) function of only $S_{\alpha*}$
as far as independence analysis is concerned.

Now, we describe the binary split procedure.  
To split any $S_{\alpha*}$, we simply generate an RV $L_{\alpha*}$ using a conditional distribution that we will specify next, and then let 
$S_{\alpha0*} \triangleq L_{\alpha*}$ and $S_{\alpha1*} \triangleq S_{\alpha*} - L_{\alpha*}$.
Since the split procedure is the same for any $\alpha*$, 
we drop the subscript $\alpha*$ from $S_{\alpha*}$ and $L_{\alpha*}$ in describing it whenever possible.
In the following derivations and proofs, we assume that $S$ is a continuous RV, so its probability density function (pdf) is used;  
if $S$ is instead a discrete RV, we can use its probability mass function (pmf) instead.
To split $S$ {\it for the first time}, {\it a fresh RV} $L$ is generated according to the following conditional pdf:
\begin{equation}\label{splitpost}
 f_{L|S}(L = x | S = z) \triangleq \rho_n(x)\rho_n(z-x) / \rho_{2n}(z),
\end{equation} 
where  $\rho_n(\cdot)$ and $\rho_{2n}(\cdot)$ are the pdfs of $X^{*n}$ and $X^{*2n}$ respectively. 
For notational simplicity, we drop the subscript $L|S$ from $f_{L|S}$ in the sequel. 
The following theorem states that this split procedure satisfies the aforementioned property (I).

\begin{theorem}\label{th:split}
If $S$ has distribution $X^{*2n}$, then the conditional distribution of $L|S$ in \autoref{splitpost} implies that $L$ and $S-L$ are i.i.d. RVs having distribution $X^{*n}$.
\end{theorem}
\begin{proof}
	We first calculate the joint pdf of $L$ and $S-L$ as follows
\begin{align}
	f(L=x, S-L=v) 
	&= f(L=x | S=x+v)f(S=x+v)  
	=  \rho_n(x)\rho_n(v) \label{line:rhoxrhoy}, 
\end{align}
where 
\autoref{line:rhoxrhoy} can be derived from \autoref{splitpost} by replacing $z$ with $x+v$.


Hence we have 
$
	f(L=x) = \int_{-\infty}^{\infty} \rho_n(x)\rho_n(v) \mathrm{d}v =  \rho_n(x).
$
Similarly, $f(S-L=v) = \rho_n(v)$.  Hence we have $f(L=x, S-L=v) = f(L=x) f(S-L=v)$, which proves the independence.
\end{proof}



We now put the index subscript $\alpha*$ back into $S$ and $L$, since we need to state results concerning a set of $S$ and $L$ terms with different indices.  
We pause to clarify the mathematical meanings of two emphasized phrases used in stating the split procedure.
The first phrase is ``for the first time''. 
It means that, 
in case $S_{\alpha*}$ is to be split again, the same $L_{\alpha*}$, that was generated and used for the first time, must be used again.  
This is a basic requirement for generating RVs, because the values of RVs should be fixed upon generation, which is not even a part of the problem statement. 
The second phrase is ``a fresh RV''.  It means that each $L_{\alpha*}$ is generated based on only the value of $S_{\alpha*}$ using fresh 
randomness.  
As a result, the random vector $\langle S_{\alpha0*}, S_{\alpha1*}\rangle$ is a vector function of only 
$S_{\alpha*}$ as far as independence analysis is concerned, which 
is precisely 
property (II).  The language of property (II), such as ``fresh randomness'', is a bit vague right now.
It will be further simplified and clarified in~\autoref{sec:erssolutions}.

The aforementioned idealized assumption is simply 
that we can 
somehow remember the fresh randomness involved in generating each $L_{\alpha*}$ (for the first time), so that property (II) can be ensured.
However, since the number of non-leaf prefixes in each DST is $O(U)$, it is typically prohibitively expensive in terms of storage cost to remember 
such fresh randomness for every  
$L_{\alpha*}$ generated, and this idealized assumption is impractical.   Since property (II) depends on this assumption, it is also impractical.  
In~\autoref{sec:dstimpl}, we will introduce a slightly weakened property (II*) that does not require this assumption, yet can still lead to strong provable
statistical guarantees.

Before we state and prove the following theorem, we introduce a third notation $S^l_{(i*)}$ for a dyadic range-sum (besides $S[a, b)$ and $S_{\alpha*}$).
$S^l_{(i*)}$ represents the same dyadic range-sum as $S_{\alpha*}$,
if the number $i$, written as an $l$-bit binary number, is (the binary prefix) $\alpha$.  
For example, in the example shown in \autoref{fig:prefix}, $S^2_{(1*)}$ is equivalent to 
$S_{01*}$ and $S[4, 8)$.  We define $L^l_{(i*)}$ similarly (as the $L$ involved in splitting $S^l_{(i*)}$).
Note that if $S_{\alpha*}$ is the same as $S^l_{(i*)}$,
then $S_{\alpha0*}$ and $S_{\alpha1*}$, the two children of $S_{\alpha*}$, are the same as $S^{l+1}_{((2i)*)}$ and $S^{l+1}_{((2i+1)*)}$ respectively. 

Since the DST is a complete binary tree, there are $2^l$ nodes at the $l^{th}$ level down the root.  
Under this $S^l_{(i*)}$ notation, these $2^l$ nodes are $S^l_{(0*)}$, $S^l_{(1*)}$, $\cdots$, $S^l_{(\lambda_l*)}$, where $\lambda_l = 2^l-1$ (defined 
for any $l$).
The following theorem states that for any $1 \le l \le \log_2 U$, these $2^l$ dyadic range-sums are i.i.d. RVs.  

\begin{theorem}\label{th:dyasim}
Suppose that the split procedure satisfies  properties (I) and (II).  
Then, for any $l$, $1\le l \le  \log_2 U$, the $2^l$ dyadic range-sums $S^l_{(0*)}$, $S^l_{(1*)}$, $\cdots$, $S^l_{(\lambda_l*)}$ at level $l$ 
have i.i.d.  distribution $X^{*(U/2^l)}$.

\end{theorem}
\begin{proof}
	We prove by induction on $l$. 
	For the base case when $l=1$, there are two dyadic range-sums at the $1^{st}$ level: $S^1_{(0*)}$ and $S^1_{(0*)}$.
Since they result from splitting $S_*$, which has distribution $X^{*U}$ (by initialization), $S^1_{(0*)}$ and $S^1_{(0*)}$ are i.i.d. RVs  with distribution $X^{*(U/2)}$ according to property (I).

	
	Now, we prove the case of $l+1$ from that of $l$. By the induction assumption, for any $i$, the parent $S^l_{(i*)}$  follows $X^{*(U/2^l)}$, so by property (I), its two children $S^{l+1}_{((2i)*)}$ and $S^{l+1}_{((2i+1)*)}$ are independent and each has the marginal distribution $X^{*(U/2^{l+1})}$.
	We denote this as \emph{fact (*)}.
	It remains to show $S^{l+1}_{(0*)}$, $S^{l+1}_{(1*)}$, $\cdots$, $S^{l+1}_{(\lambda_{l+1}*)}$, the generated range-sums on level $l+1$, are  independent. By induction assumption, $S^l_{(0*)}$, $S^l_{(1*)}$, $\cdots$, $S^l_{(\lambda_l*)}$ are  independent.
Each $\langle S^{l+1}_{((2i)*)}, S^{l+1}_{((2i+1)*)}\rangle$ is a (vector) function of only $S^l_{(i*)}$, which we call property (II) earlier. 
	Hence the random vectors $\langle S^{l+1}_{((2i)*)}, S^{l+1}_{((2i+1)*)}\rangle$ are  independent  for different $i$, which we denote as \emph{fact (**)}.

	Therefore, the independence of all values on level $l+1$ follows from the following factorization of the joint cdf for any sequence of values $x_0$, $x_1$, $\cdots$, $x_{\lambda_{l+1}}\in\mathbb{R}$, \newline
	$\Pr\left(S^{l+1}_{(0*)} \leq x_0,   S^{l+1}_{(1*)} \leq x_1,\cdots, S^{l+1}_{(\lambda_{l+1}*)} \leq x_{\lambda_{l+1}} \right)  \newline=  \prod_{i=0}^{\lambda_l}\Pr\left(S^{l+1}_{((2i)*)}\leq x_{2i}, S^{l+1}_{((2i+1)*)} \leq x_{2i+1}\right) \newline=  \prod_{i=0}^{\lambda_l}\Pr\left(S^{l+1}_{((2i)*)} \leq x_{2i}\right)
	\Pr\left( S^{l+1}_{((2i+1)*)} \leq x_{2i+1}\right) = \prod_{i=0}^{\lambda_{l+1}} \Pr\left(S^{l+1}_{(i*)} \leq x_{i}\right)$, 
where the first equation is due to fact (**) above 
and the second is due to fact (*) above.
\end{proof}

\begin{corollary}\label{cor:xiindep}
	The underlying RVs $X_0, X_1, \cdots, X_{U-1}$, which are $S^l_{(0*)}$, $S^l_{(1*)}$, $\cdots$, $S^l_{(\lambda_l*)}$ for $l=\log_2 U$, have i.i.d. distribution $X$.
\end{corollary}

\begin{remark*}The following observation, which is a part of fact (*) in the proof above of~\autoref{th:dyasim}, continues to hold when property (II) is taken away, since the proof of this part only needs property (I).\end{remark*}

\begin{observation}\label{fact:marginal}
	Even if the split procedure satisfies only property (I), each $S^l_{(i*)}$ still has marginal distribution $X^{*(U/2^l)}$.
\end{observation}

The logic of the induction step in the proof of~\autoref{th:dyasim} can be stated as the following lemma, which will be used in the proofs 
in~\autoref{sec:dstimpl}.


\begin{lemma}\label{lem:induction} If a set of $k > 1$ distinct dyadic range-sums $S^l_{(i_1*)}$, $S^l_{(i_2*)}, \cdots$, $S^l_{(i_k*)}$ at level $l$ are independent and they 
are \emph{split conditionally independently}, then their $2k$ children $S^{l+1}_{((2i_1)*)}$, $S^{l+1}_{((2i_1+1)*)}$, $S^{l+1}_{((2i_2)*)}$, $S^{l+1}_{((2i_2+1)*)}$, $\cdots$, $S^{l+1}_{((2i_k)*)}$, $S^{l+1}_{((2i_k+1)*)}$ at level
$l+1$ are also independent.  
\end{lemma}

\begin{remark*}
Here, ``split conditionally independently'' means the following two conditions that together lead to fact (**). First, the RVs involved in these splits, namely $L^l_{(i_1*)}$, $L^l_{(i_2*)}, \cdots$, $L^l_{(i_k*)}$ are (conditionally) independent provided that 
$S^l_{(i_1*)}$, $S^l_{(i_2*)}, \cdots$, $S^l_{(i_k*)}$ are independent.  Second, each such $L^l_{(i*)}$ involved is a (randomized) function of $S^l_{(i*)}$ only.
\end{remark*}

\subsection{Efficient Range-Summable (ERS) Solutions} \label{sec:erssolutions}



As explained earlier, every DST-based solution boils down to generating $L_{\alpha*}$  according to the conditional distribution $f( L_{\alpha*} | S_{\alpha*} )$ specified in \autoref{splitpost}.  
Although \autoref{splitpost} applies to any distribution $X$ in principle, for such a solution to work, two 
hurdles have to be overcome.  
The first hurdle is a mathematical one:  Nice closed-form formulae for $\rho_n(x)$ (pdf of $X^{*n}$) and $\rho_{2n}(x)$ (pdf of $X^{*2n}$), and hence for $f( L_{\alpha*} | S_{\alpha*} )$, 
appear to exist for only a few such ($X$)'s.  For other target distributions, designing DST-based ERS solutions appears to be challenging.  

Even when the distribution $X$ is nice so that we have a closed-form formula, we are still facing the second hurdle, which is to generate $L_{\alpha*}$ in a computationally efficient manner.  
A computational procedure for generating $L_{\alpha*}$ is typically a two-step process as follows.  
First, we generate a {\it fresh} (i.e., independent of all other RVs including $S_{\alpha*}$) uniform random $\mu$-bit-long binary string $C_{\alpha*}$ that, if viewed as nonnegative integer, is uniformly distributed in the set $\{0, 1, 2, \cdots, 2^\mu-1\}$.   
Usually $\mu = 32$ provides enough statistical precision.
Second, $L_{\alpha*}$ is set to $\theta(C_{\alpha*}, S_{\alpha*})$, where $\theta(x, z)$ is a {\it deterministic} function designed in such a way 
that the resulting $L_{\alpha*}$ has the right conditional distribution as specified in \autoref{splitpost}.

Now we are ready to simplify the language of property (II) as promised earlier.  
The simplified property (II) is that each $C_{\alpha*}$ is a fresh RV (that is independent of any other RV). 
As a result, 
each $L_{\alpha*} = \theta(C_{\alpha*}, S_{\alpha*})$ is a {\it fresh} RV that is a function only of $S_{\alpha*}$, which is precisely property (II). 
With this simplified property (II), the idealized assumption becomes that each such $C_{\alpha*}$ (not $L_{\alpha*}$) needs to be remembered after it is first generated.  


In probability theory, the standard textbook technique, called inverse transform method, is to let $\theta(x, z) = F^{-1}(x | z)$ where $F(x | z) \triangleq \int_{-\infty}^x f(v| z)\mathrm{d}v$ is the conditional cdf of $ L_{\alpha*} | S_{\alpha*}$.
However, inverse transform is usually not computationally efficient, since the inverse conditional cdf $F^{-1}(x | z)$ usually does not have a closed form, as we will elaborate in \autoref{sec:cauchy}.  
We will show that, for all three ERS solutions, we propose alternative designs of $\theta(x, z)$ 
that are much more efficient, in terms of computational and/or space complexity, than the respective inverse transforms.
Finally, 
when $X$ is a discrete RV (e.g., when $X$ is a single-step random walk), it is possible to precompute $F^{-1}(x | z)$ for all possible values of $x$ and $z$, 
and store the values in a table.  This technique, called the tabular inverse transform~\cite{marsaglia}, can only be used when the memory cost of 
storing the table is manageable.

\subsection{Gaussian-DST}\label{sec:gauss}

For notational simplicity, we again drop the subscript $\alpha*$ from $S_{\alpha*}$, $L_{\alpha*}$, and $C_{\alpha*}$ in describing the binary split procedures in the sequel.
When $X$ is standard Gaussian $\mathcal{N}(0, 1)$, $X^{*n}$ is $\mathcal{N}(0, n)$ with pdf $\rho_n(x) = 1/\sqrt{2\pi n}\cdot \exp(-x^2/(2n))$, and $X^{*2n}$ is $\mathcal{N}(0, 2n)$ with pdf $\rho_{2n}(x) = 1/\sqrt{4\pi n}\cdot \exp(-x^2/(4n))$.  
According to \autoref{splitpost}, we have $f(L=x|S=z) = \rho_n(x)\rho_n(z-x) / \rho_{2n}(z) = 1/\sqrt{\pi n}\cdot \exp(-(x - z/2)^2/n)$, which can be written as the pdf of
$\mathcal{N}(z/2, n/2)$.   We generate $L$ according to the (value of) random string $C$ from the distribution $\mathcal{N}(z/2, n/2)$ as follows.
$L$ is set to $z/2 +Y$, where $Y$ 
is a fresh Gaussian RV with distribution $\mathcal{N}(0, n/2)$ generated from $C$ using efficient techniques such as 
Box-Muller transform~\cite{boxmuller}.  In~\cite{dyahist}, no specific technique was suggested for generating this $L$.

\subsection{Cauchy-DST}\label{sec:cauchy}

Now we describe how to generate $L$ from $C$ when the target distribution $X$ is standard $\mathrm{Cauchy}(0, 1)$.
By the stability property of Cauchy distribution, the $n^{th}$ convolution power 
$X^{*n}$ is $\mathrm{Cauchy}(0, n)$, which 
has pdf
$
	\rho_n(x) = \left(\pi n\left[1+\left(x/n\right)^2\right]\right)^{-1}
$.
The pdf of $X^{*2n}$ is
$
	\rho_{2n}(x) = \left(2\pi n\left[1+\left(x/2n\right)^2\right]\right)^{-1}
$. 
 Therefore, 
 by \autoref{splitpost}, the conditional pdf 
\begin{equation}\label{eq:cauchysplit}
	\begin{split}
f(L=x|S=z) = \frac{\rho_n(x)\rho_n(z-x)}{\rho_{2n}(z)}= \frac{n}{2\pi}\cdot\frac{z^2+4n^2}{(n^2+
	x^2)(n^2+(z-x)^2)}.
\end{split}\end{equation}

In~\cite{dyahist}, it was suggested that the inverse transform method described above be used to generate $L$.
The rationale offered in~\cite{dyahist} was that since the conditional pdf $f(x|z)$ in \autoref{eq:cauchysplit} is a rational fraction, the conditional cdf $F(x|z)$ has a closed-form expression~\cite{stewart2020calculus}, which makes its inverse $F^{-1}(x|z)$ numerically calculable.
However, 
the procedure for calculating $F^{-1}(x|z)$ has a high computational complexity in practice, since the (closed-form) formula of $F(x|z)$ is very complicated.

We propose a much more efficient way of sampling $L$ from $f(x|z)$ based on a Monte-Carlo simulation technique 
called \emph{rejection sampling}~\cite{rejsample}.
The idea of rejection sampling is that, 
we instead sample another RV $Y$ from another pdf $\psi(x|z)$ that is computationally easier to sample from than $f(x|z)$.
Supposing the value of this sample is $x$.   Then this sample is accepted 
with probability $\gamma = f(x|z) / (Q \psi(x|z))$ and rejected with probability $1 - \gamma$.
The rejection sampling step is 
repeated until a sample (of $Y$) is accepted, and the finally accepted sample is (the realized value of) $L$.
Here, this constant $Q$ should be set such that $\gamma \le 1$ for all values of $x$ and $z$, or in other words
$Q \geq \max_{x, z} f(x|z) / \psi(x|z)$. 
In statistics, a key objective as well as challenge in designing a rejection sampling procedure is to select $\psi(x|z)$ so that 
$\max_{x, z} f(x|z) / \psi(x|z)$ and hence this $Q$ can be made as small as possible (since the accept probability
can be as small as $1/Q$).  







The $\psi(x|z)$ we propose is the pdf of the following mixture RV $Y$:  $Y$ is equal to $Y'$ or $Y' + z$ each with probability $1/2$ (depending on 
the value of $C$), where $Y'$ is a fresh RV with distribution $\mathrm{Cauchy}(0, n)$.  This $Y'$ can generated from $C$ via the aforementioned
inverse transform $Y' =F^{-1}_{Y'}(C)= n\tan(\pi(C-1/2))$;  note that, unlike the conditional inverse cdf $F^{-1}(x|z)$ described above, the unconditional inverse cdf 
$F^{-1}_{Y'}(C)$ here takes a much simpler form and hence
can be computed efficiently.
It can be shown that the conditional pdf of $Y$ is 
\begin{displaymath}\label{causp}\begin{split}
	\psi(L=x|S=z) =\frac{\rho_n(x) + \rho_n(x-z)}{2}=	\frac{n}{2\pi}\cdot\frac{2n^2+x^2+(z-x)^2}{(n^2+x^2)(n^2+(z-x)^2)}.
\end{split}\end{displaymath}
We set the parameter $Q$ to $2$ since for any $x$ and $z$,
\begin{equation}\begin{split}\label{eq:accratio}
		\frac{f(L=x|S=z)}{\psi(L=x|S=z)} = \frac{4n^2 + z^2}{2n^2+x^2+(z-x)^2} = \frac{4n^2 + z^2}{2n^2+z^2/2+2(x-z/2)^2} \leq 2,
	\end{split}
\end{equation}
so the average accept probability of a sample ($Y$) is at least $1/2$.


\subsection{Random Walk (RW)-DST}\label{sec:randw}

We now describe how to generate $L$ from $S = z$ and $C$ where the target distribution $X$ is a single-step random walk.
We first derive $f(L=x|S=z)$.  
Since $X^{*n}$ has pmf $\rho_n(x) = 2^{-n}\binom{n}{(n-x)/2}$, and  $X^{*2n}$ has  pmf $\rho_{2n}(x) = 2^{-2n}\binom{2n}{(2n-x)/2}$, 
by \autoref{splitpost}, the conditional pmf 
\begin{equation}\label{eq:rwsplit}
	\begin{split}
f(L=x|S=z) =\frac{ \rho_n(x)\rho_n(z-x)} {\rho_{2n}(z)} =\binom{n}{(n-x)/ 2} \binom{n}{(n-z+x)/2}  \bigg/ \binom{2n}{n-z/2},
\end{split}\end{equation}
if $z$ is an even integer such that $-2n\leq z\leq 2n$, $x$ is an integer such that $-n\leq x\leq n$ and $-n+z\leq x \leq n+z$, and $n-x$ is even;
otherwise $f(L=x|S=z) = 0$.


 

We now introduce a concept that will become handy in the rest of this section.  
We say that $y$ is a {\it probable value} of a discrete RV $Y$, if
the probability $P(Y = y)$ is not vanishingly small or $0$.  
This concept is important here, because we will trade memory space for computation time by precomputing 
and storing some conditional probability values, and the memory cost could be greatly reduced if we store 
only those for probable values of $S$ and $L$ (conditioned upon $S$).
Now we analyze the asymptotic number of probable values of $S$ and $L$ when $n$ is a large number.  
For $S$, only integers that are no larger than $O(\sqrt{n})$ are probable, since its pmf $X^{*2n}$
converges to $\mathcal{N}(0, 2n)$ by the central limit theorem.  
Hence, by storing probability values only for the probable values of $S$, the space complexity reduces to $O(\sqrt{n})$ from $O(n)$.
The same can be said about $L$ 
since its conditional pmf can be shown to converge to $\mathcal{N}(z/2, n/2)$ in~\autoref{th:sqrt2ratio}.

We had tried the aforementioned tabular inverse transform method~\cite{marsaglia} 
on $f(L=x| S=z)$.
However, even when this probable value trick is used, the memory cost is still very high for most applications.
The total memory cost is $O(U)$, since for each of the $\log U$ values of $n$, we need to store the values of $f(L=x | S=z)$ for all combinations of 
$O(\sqrt{n})$ probable $z$ values and $O(\sqrt{n})$ probable $x$ values, and the largest $n$ value is $U$.  
For example, when $U = 2^{20}$,
the total size of the precomputed tables would still be several gigabytes. 

We propose a rejection sampling technique that, in combination with the tabular inverse transform and the probable value trick,
provides a fast, space-efficient, and accurate solution to this ERS problem. 
Like in \autoref{sec:cauchy}, the rejection sampling method is specified by the RV $Y$ whose conditional pdf (given $S=z$) is $\psi(x|z) = 2^{-n}\binom{n}{(n-x+2\lceil z/4\rceil)/2}$, and the constant $Q$ (defined later). 
$Y$ can be generated as $Y' + 2\lceil z/4\rceil$, where $Y'$ is a fresh RV with distribution $X^{*n}$ generated from $C$ by tabular inverse transform~\cite{marsaglia}.
Our next step is to determine $Q$ (as an upper bound on $\max_{x, z} f(x|z) / \psi(x|z)$) 
for each $n$ value and for all probable $z$ values (those that are $O(\sqrt{n}$) as explained 
earlier).
For all $n \geq 256$, we know from calculations and from \autoref{th:sqrt2ratio} that this maximum value is at most $1.47$. Hence, we set $Q = 1.47$ so that the average accept probability is at least $1/1.47 = 0.68$.
The rejection sampling operation is computationally efficient, because both $f(x|z)$ and $\psi(x|z)$ can be computed in $O(1)$ time if the factorials $i!$ and $(n-i)!$ 
are precomputed for probable $i$ values (that is $i = O(\sqrt{n})$).  
When $n\geq 256$, we use rejection sampling (with $Q = 1.47$).  When $n\leq 128$, we use the tabular inverse transform (with the probable value trick) 
since the table size grows as $O(n)$ as explained earlier.  When $U = 2^{20}$ like in the example above, the total size of the precomputed tables (for all 
20 values of $n$) is only several megabytes.


\begin{theorem}\label{th:sqrt2ratio}
The maximum ratio $\max_{x=O(\sqrt{n})}f(x|z) / \psi(x|z)$ converges to $\sqrt{2} \approx 1.414$ when $n\to\infty$ and $z=O(\sqrt{n})$.
\end{theorem}

 \begin{proof}
By de Moivre-Laplace Theorem~\cite{demoivre}, when $n\to\infty$ and $x = O(\sqrt{n})$ is a probable value, $\rho_n(x)$ converges to the pdf of $\mathcal{N}(0, n)$, which is $1/\sqrt{2\pi n}\cdot \exp(-x^2/(2n))$. As a result, $f(x|z)$ in \autoref{eq:rwsplit} converges to $1/\sqrt{\pi n}\cdot \exp(-(x - z/2)^2/n)$, the conditional pdf for splitting Gaussian RVs. Meanwhile, $\psi(x|z)$, the conditional pdf of $Y$ converges to $1/\sqrt{2\pi n}\cdot \exp(-(x-2\lceil z/4\rceil)^2/(2n))$.
If $z$ is a multiple of $4$, the maximum ratio is achieved on $x=z/2$, and the ratio is $\sqrt{2}$. Otherwise, $z$ is an even number but not a multiple of $4$, the maximum ratio is achieved on $x=2\lfloor z/4 \rfloor$, and the ratio is $\sqrt{2}\exp(1/n)$, which converges to $\sqrt{2}$ when $n\to\infty$.
\end{proof}


\subsection{Speed of Dyadic Simulation}\label{sec:dyaperf}

Recall that our idealized and impractical assumption is that 
we can somehow remember the value of every $C_{\alpha*}$ after it was first generated, with which
we can rigorously prove that $X_0, X_1, \cdots, X_{U-1}$ are i.i.d.
As mentioned in~\autoref{sec:indepguarantee}, this assumption 
can be removed by
instead computing each such $C_{\alpha*}$ as $h(\alpha)$, where $h(\cdot)$ is a hash function.

We have implemented the DST framework using 
an off-the-shelf hash function family called {\tt wyhash}~\cite{wyhash}.  
{\tt wyhash} offers two attractive advantages.
First, computationally {\tt wyhash} is very efficient:  It takes roughly two nanoseconds for {\tt wyhash} to compute a hash value~\cite{wyhash}.
Second, it guarantees excellent empirical independence among the values of ($C_{\alpha*}$)'s generated~\cite{smhasher}.  
To further improve this empirical independence, we use a different (independent) hash function at each level of the DST.
The storage cost of a DST is tiny, since each hash function uses only a $32$-bit random seed that needs to be remembered.
\autoref{tab:perf} shows the average amount of time it takes for a DST to split a Gaussian, Cauchy, and random walk RV respectively, measured on a workstation running Ubuntu 18.04 with 
Intel(R) Core(TM) i9-10980XE \SI{3.00}{GHz} CPU.
It is a few times faster to split a Gaussian than to split the other two, largely because the other two involve rejection sampling, which is 
a relatively computationally intensive process.  

\begin{table}[!ht]
	\caption{Average split time of an RV.}\label{tab:perf}
	\centering
	{\begin{tabular}{|l|c|c|c|}
			\specialrule{.1em}{.05em}{.05em} 
			Distribution &  Gaussian & Cauchy & Random Walk\\ \hline
			Time per split (\SI{}{ns})& 4.8 & 24.8 & 21.2\\
			\specialrule{.1em}{.05em}{.05em} 

	\end{tabular}}
\end{table}
\section{$k$-wise Independence Theory for DST}\label{sec:dstimpl}

At the end of the previous section, we have shown that, by using a per-level {\tt wyhash} function to hash a prefix $\alpha*$ into a uniform random string $C_{\alpha*}$,
our solutions have high performance and the underlying RVs are empirically independent.  However, {\tt wyhash} does not provide any theoretical guarantee concerning
independence.  In this section, we describe our novel $k$-wise independence theory that provides both high computational efficiency and strong provable
independence guarantees. Our $k$-wise theory for DST is motivated by the fact that most ERS applications do not require the underlying RVs to be all-wise independent.  For example, it can be 
shown (e.g., using arguments similar to those in Theorem 2.2 in~\cite{ams}) that, that for the $L_2$-norm estimation, 
when the estimator is the standard quadratic polynomial (of the accumulators) $\hat{d}^2_2  = (A^2_1 + A^2_2 + \cdots + A^2_r)/r$, 
the streaming algorithm described in~\autoref{sec:application} achieves the same statistical efficiency whether the underlying Gaussian RVs are $4$-wise independent or all-wise independent.

 To this end, our idea is to use $\log_2 U$ $k$-wise independent hash functions (instead of {\tt wyhash}).
A $k$-wise independent hash function $h(\cdot)$ has the following property:  Given an arbitrary set of $k$ 
different keys $i_1, i_2, \cdots, i_k$, their hash values $h(i_1), h(i_2), \cdots, h(i_k)$ are independent.   
Such hash functions are very computationally efficient when $k$ is a small number
such as $k=2$ (roughly 2 nanoseconds per hash just like {\tt wyhash}) and $k = 4$ (several nanoseconds per 
hash)~\cite{univhash,tabulation4wise,tabulationpower}.

Like in the earlier case (of using {\tt wyhash}), a different per-level hash function $h^l(\cdot)$, that is $k$-wise independent, 
is used at each level $l$ of the DST, and each random string $C^l_{(i*)}$ (used to generate $L^l_{(i*)}$) is hash-generated as $ h^l(i)$.
This construction weakens property (II) slightly.
The weakened one, called \emph{property (II*)}, is that, at any level $l$, any $k$ distinct range-sums $S^l_{(i_1*)}, S^l_{(i_2*)}, \cdots, S^l_{(i_k*)}$ 
are split conditionally independently.  
This construction can guarantee property (II*), 
because their ``split seeds'' $C^l_{(i_1*)}, C^l_{(i_2*)}, \cdots, C^l_{(i_k*)}$  are 
not only independent among themselves (thanks to $h^l(\cdot)$ being $k$-wise independent) but also 
independent of $S^l_{(i_1*)}, S^l_{(i_2*)}, \cdots, S^l_{(i_k*)}$  (since $h^l(\cdot)$ is a fresh hash function that has never been used in 
hash-generating any such $S^l_{(i*)}$).   
With this construction, the DST has the following nice $k$-wise independence property at every level.  

\begin{theorem}\label{th:nwise}
	If every $h^l(\cdot)$, $1 \leq l < \log_2 U$, is $k$-wise independent, then for any $l$, $1 \leq l \leq \log_2 U$, 
	the $2^l$ range-sums $S^l_{(0*)}$, $S^l_{(1*)}$, $\cdots$, $S^l_{(\lambda_l*)}$ are $k$-wise independent. 
\end{theorem}
\begin{proof}
The proof is similar to that of \autoref{th:dyasim} by induction. 
For the base case when $l=1$, there are two dyadic range-sums at the $1^{st}$ level: $S^1_{(0*)}$ and $S^1_{(1*)}$. Since they result from splitting $S_*$, which follows $X^{*U}$ (by initialization), $S^1_{(0*)}$ and $S^1_{(1*)}$ are i.i.d. RVs according to property (I).




Now, we prove the induction on level $l+1$ from level $l$. 
For any fixed set of $k$ indices $i_1$, $i_2$, $\cdots$, $i_k$ on level $l+1$, we need to prove $S^{l+1}_{(i_1*)}$, $S^{l+1}_{(i_2*)}$, $\cdots$, $S^{l+1}_{(i_k*)}$ are independent.  This follows from \autoref{lem:induction}, since these $k$ elements are the children of no more than $k$ parents after duplicates are removed. These parents, no more than $k$ in number, are
independent by the induction hypothesis and are independently split by property (II*).  
\end{proof}

The theorem above implies that the underlying RVs $ X_0, X_1, \cdots, X_{U-1}$, which are the $U$ singleton range-sums at level $\log_2 U$, are also $k$-wise independent. 
The following theorem is a surprising result, since although it requires only $2$-wise independence, it provides a 
very strong statistical guarantee.  As mentioned earlier in~\autoref{sec:application}, when $X$ is a single-step random walk, this guarantee satisfies the requirement of MP-RW-LSH.
The following theorem is an immediately corollary of the lemma that follows it.

\begin{theorem}\label{cor:stdist}
 If every $h^l(\cdot)$, $1\leq l < \log_2 U$, is $2$-wise independent, then for any integers $a$ and  $b$ such that $0\leq a \leq b < U$, the range-sum $S[a, b)$ 
has marginal distribution $X^{*(b-a)}$. 
\end{theorem}

\begin{lemma}\label{th:twowise}
	If every $h^l(\cdot)$ is $2$-wise independent, then	for any $1\leq l <\log_2 U$ and  $0 \leq a, b\leq U$ such that $a \leq b$, then the following two properties hold.
	\begin{enumerate}
	\item The three RVs $\sum_{i=a+1}^{b-1} S^l_{(i*)}$, $S^l_{(a*)}$, and $S^l_{(b*)}$ are  independent.
	\item The range-sum $S[(a+1)U/2^l, bU/2^l) = \sum_{i=a+1}^{b-1} S^l_{(i*)}$  
	follows distribution $X^{*(U(b-a-1)/2^l)}$, where $U(b-a-1)/2^l$ is the number of underlying RVs contained in the range $[(a+1)U/2^l, bU/2^l)$.
	\end{enumerate}
\end{lemma}
%
Before we start the proof, we note that \autoref{fact:marginal}, which states that 
the marginal distribution of any range-sum $S^l_{(i*)}$ is $X^{*(U/2^l)}$, continues to hold despite the weakening of property (II) in this section.  
\begin{proof}	
	The proof for \autoref{th:twowise} is by induction on $l$. For the base case when $l=1$, the three RVs, after 0's and duplicates are removed, belong to the set of 
	two range-sums on the first level, $S^1_{(0*)}$ and $S^1_{(1*)}$.  These two range-sums  
	have distribution i.i.d. $X^{*(U/2)}$ thanks to property (I). This leads to the two properties on the first level. 
	
	We now prove the case of level $l+1$ from that of level $l$. 
	We first define the following  notations: $a'\triangleq \lfloor a/2\rfloor$, so $S^l_{(a'*)}$ is the parent of $S^{l+1}_{(a*)}$;  $b'\triangleq \lfloor b/2\rfloor$; and $\tilde{i}$ is defined as $i+1$ if $i$ is even (the younger of the siblings), and as $i-1$ otherwise, so $S^l_{(\tilde{i}*)}$ is always the other sibling of $S^l_{(i*)}$.
	The induction claim of the first property that $\sum_{i=a+1}^{b-1} S^{l+1}_{(i*)}$ and $\langle S^{l+1}_{(a*)}, S^{l+1}_{(b*)}\rangle$ are independent holds, due to the following two facts: (i) $\sum_{i=a'+1}^{b'-1} S^l_{(i*)}, S^{l+1}_{(a*)}$, $S^{l+1}_{(\tilde{a}*)}$, $S^{l+1}_{(b*)}$, and $S^{l+1}_{(\tilde{b}*)}$ are  independent; (ii) $\sum_{i=a+1}^{b-1} S^{l+1}_{(i*)}$ is a deterministic function of $\langle \sum_{i=a'+1}^{b'-1} S^l_{(i*)}, S^{l+1}_{(\tilde{a}*)}, S^{l+1}_{(\tilde{b}*)}\rangle$ in the sense that 
		$\sum_{i=a+1}^{b-1} S^{l+1}_{(i*)} = \sum_{i=a'+1}^{b'-1} S^l_{(i*)} + S^{l+1}_{(\tilde{a}*)} \mathbf{1}_{even}(a)+ S^{l+1}_{(\tilde{b}*)} \mathbf{1}_{odd}(b)$, where the indicator function $\mathbf{1}_{even}(a)$ is $1$ if $a$ is an even integer and is $0$ otherwise, and $\mathbf{1}_{odd}(b)$ is similarly defined.
	
We now prove fact (i). By the induction assumption, the three RVs $\sum_{i=a'+1}^{b'-1} S^l_{(i*)}$, $S^l_{(a'*)}$, and $S^l_{(b'*)}$ are  independent. Since $h^{l}(\cdot)$ is a fresh $2$-wise independent hash function, property (II*) holds for $k=2$, so the four children $S^{l+1}_{(a*)}$, $S^{l+1}_{(\tilde{a}*)}$, $S^{l+1}_{(b*)}$, and $S^{l+1}_{(\tilde{b}*)}$ are  independent by \autoref{lem:induction}. 
Furthermore, these four children and $\sum_{i=a'+1}^{b'-1} S^l_{(i*)}$ together 
	are  independent, because by property (II*), $\langle S^{l+1}_{(a*)}, S^{l+1}_{(\tilde{a}*)}\rangle$ is a function of only $S^l_{(a'*)}$ but not $\sum_{i=a'+1}^{b'-1} S^l_{(i*)}$, which is composed of the other range-sums, and similarly $\langle S^{l+1}_{(b*)}, S^{l+1}_{(\tilde{b}*)}\rangle$ is a function of only $S^l_{(b'*)}$. By the induction assumption, $\sum_{i=a'+1}^{b'-1} S^l_{(i*)}$ is independent of these four children.

	The induction claim of the second property that $\sum_{i=a+1}^{b-1} S^{l+1}_{(i*)}$ has distribution $X^{*n}$ with $n = (b-a-1)U/2^{l+1}$ holds, because of the induction assumption that $\sum_{i=a'+1}^{b'-1} S^l_{(i*)}$ has distribution $X^{*n'}$ where $n' = (b'-a'-1)U/2^{l}$ is the number of $(X_i)'s$ (underlying RVs) contained in its range $[(a'+1)U/2^{l}, b'U/2^{l})$. To see why this claim holds, we have to go through the four possible cases on the parities of $a$ and $b$. We show the most inclusive case where $a$ is even and $b$ is odd, and the other three cases are just similar. In this case, $b-a = 2(b'-a')+1$, so $n = n'+U/2^l$. $\sum_{i=a+1}^{b-1} S^{l+1}_{(i*)}$ follows $X^{*n}$, because it is the sum of the following three independent RVs, $\sum_{i=a'+1}^{b'-1} S^l_{(i*)}$, which follows $X^{*n'}$, $S^{l+1}_{(\tilde{a}*)}$, which follows $X^{*(U/2^{l+1})}$ by \autoref{fact:marginal}, and $S^{l+1}_{(\tilde{b}*)}$, which also follows $X^{*(U/2^{l+1})}$.
	


\end{proof}

\begin{remark*}
 In this proof, to compute the range-sum $S[a, b) = \sum_{i=a}^{b-1} S^{\log_2 U}_{(i*)}$, at most two splits need to performed at each level $0 \leq l < \log_2 U$, on namely $S^{l}_{(a_l*)}$ and $S^{l}_{(b_l*)}$ (they can be the same node), where $a_l = \lfloor a 2^l/U\rfloor$ and $b_l = \lfloor b 2^l/U\rfloor$.  This implies \autoref{cor:2logntime}.    
\end{remark*}

\section{Related Work}

Since the contribution of this work is a new and practical solution approach to the ERS problem, we focus only on related works on ERS.
The ERS problem was first formulated in~\cite{feigenbaum}.
In~\cite{feigenbaum}, the aforementioned EH3, which is the first ERS solution, was proposed to augment
the AMS sketching~\cite{ams} technique.  The EH3-augmented AMS solves a wide range of new data streaming problems, 
such as estimating the size of spatial joins and the selectivity of histogram buckets, and outperforms previous ad-hoc solutions~\cite{rusu-jour}.
In the ERS literature, the one most related to this work is~\cite{dyahist}.  We have compared our work with~\cite{dyahist} in several places 
throughout this paper. 

All existing ERS solutions except~\cite{dyahist} are proposed for the case in which the target distribution $X$ is a single-step random walk.
Among them, EH3~\cite{feigenbaum} is the best known and has been compared with our dyadic simulation approach in~\autoref{sec:application}.
The ``3'' in EH3 refers to the fact that underlying RVs generated by EH3 are provably $3$-wise independent. 
BCH3~\cite{rusu-jour} is another ERS scheme that also guarantees $3$-wise independence. 
Although BCH3 is faster to compute than EH3, 
the underlying RVs generated by BCH3 are even 
more strongly correlated (beyond 4-wise)~\cite{rusu-jour} than those by EH3.  RM7~\cite{calderbank}, which guarantees $7$-wise independence, is the only existing ERS scheme 
that goes beyond $3$-wise, but it is too slow to be practical.  Empirically, RM7 takes more than $26$ milliseconds to compute a single range-sum~\cite{rusu-jour}, whereas for dyadic simulation, the time is typically less than one microsecond as shown in~\autoref{sec:dyaperf}.

Besides performance, another issues for these schemes is that they destroy all empirical independence beyond 4-wise (in the cases of EH3 and BCH3) and 8-wise (in the case of RM7).
For existing ERS solutions, this destruction (of empirical independence) is unavoidable due to the fact that 
they all solve an ERS problem by crafting a ``magic'' hash function that is based on error correction codes.  For example, in RM7 this magic hash function is defined by an 
instantiation of the Reed-Muller (RM) code.  In contrast, in our dyadic simulation approach, the ERS is achieved through a DST 
that does not require any such magic hash function:   For all practical purposes, {\tt wyhash} will do, as explained earlier.
This difference allows our approach to generalize to more target distributions and more applications.

A marginally related range-efficient computing problem to ERS, called efficient range minimizability (ERM), has been studied in the contexts of data streaming
and computational geometry.  In ERM, we would like to efficiently compute the minimum value of the RVs (that each has distribution $X$) in a range.  
An example ERM problem is when $X$ is a uniform distribution in the interval $(0, 1)$.  
We have come up with a new efficient solution to this problem, but cannot include it in this paper in the interest of space.
Any efficient solution (including ours) to this problem can be used, in combination with the MinHash sketch~\cite{minhashcounting}, to solve the 
range-efficient $F_0$ (estimation) problem~\cite{rangeeff, twoimpf0}:  
to efficiently estimate the number of distinct elements ($F_0$) in a data stream with range-updates.
Existing solutions to this problem, such as range-efficient sampling~\cite{rangeeff, twoimpf0}, are sampling-based in the sense they maintain
a select subset of sampled data items instead of a sketch (e.g., accumulators like in~\cite{stabledist}).  
The range-efficient $F_0$ problem has been generalized to high-dimensional spaces, where it is called the Klee's measure problem in computational geometry~\cite{rectangleefficient, klee2015}.  Existing solutions to Klee's measure problem are also sampling based.

\section{Conclusion}\label{sec:conclusion}
In this work, we propose \emph{dyadic simulation}, a novel solution framework to ERS that extends and improves existing frameworks in a fundamental and systematic way. We develop three novel ERS solutions for Gaussian, Cauchy, and single-step random walk distributions. We also propose novel rejection sampling techniques to make these solutions computationally efficient.
Finally, we develop a novel $k$-wise independence theory of DSTs that provide both high computational efficiency and strong provable
independence guarantees. 




\bibliographystyle{plainurl}
\bibliography{bibs/draft}

\end{document}